%% file: llncs-main.tex
\documentclass{llncs}
\usepackage[numbers,compress]{natbib}

\renewcommand*{\orcidID}[1]{\textsuperscript{\includegraphics[height=2.5ex]{orcid.pdf}}}

\usepackage[utf8]{inputenc}
\usepackage{comment}
\usepackage{enumitem}
\usepackage{url}
\usepackage{svg}
\usepackage{hyperref}
\hypersetup{colorlinks=true, linkcolor=blue, citecolor=violet, filecolor=magenta}
\usepackage{amsfonts}
\usepackage{booktabs}
\usepackage{color,xcolor}
\usepackage{mathtools}
\usepackage{multirow}
\usepackage{tcolorbox}
\usepackage[ruled,linesnumbered,noend]{algorithm2e}
\usepackage{pifont}
\usepackage{caption,subcaption}
\usepackage{graphicx}
\usepackage{balance}
\usepackage{cleveref}
\usepackage{tikz}
\usepackage{pgfplots}
\usepackage{svg}

\newcommand{\sem}[1]{[\![#1]\!]}

\usepackage[utf8]{inputenc}
\usepackage{cleveref}
\crefname{section}{§}{§§}
\Crefname{section}{§}{§§}

\SetCommentSty{mycommfont}

\newcommand{\mybox}[1]{
	\begin{tcolorbox}[
		boxsep=-2pt,
		standard jigsaw,
		boxrule=0.6pt,
		opacityback=0,
		sharp corners]
		#1
	\end{tcolorbox}
}

\newif\ifshowcomments
\showcommentstrue

\ifshowcomments
\newcommand{\yao}[1]{\textcolor{red}{[yao: #1]}}
\newcommand{\siow}[1]{\textcolor{blue}{[siow: #1]}}
\newcommand{\jiang}[1]{\textcolor{orange}{[jiang: #1]}}
\newcommand{\zuo}[1]{\textcolor{purple}{[zuo: #1]}}
\else
\newcommand{\yao}[1]{}
\newcommand{\siow}[1]{}
\newcommand{\jiang}[1]{}
\newcommand{\zuo}[1]{}
\fi

\makeatletter
\providecommand*{\input@path}{}
\newcommand{\createfiguresdir}{%
  \immediate\write18{mkdir -p figures}
}
\createfiguresdir
\makeatother

\begin{document}



\title{Shared-Context Batched Satisfiability}

\author{}
\institute{}


\author{Jiening Siow \and
Hanrui Zuo  \and
Hanyun Jiang  \and
Weiqi Wang  \and
Peisen Yao }

\institute{The State Key Laboratory of Blockchain and Data Security, Zhejiang University\\
\email{\{jiening,zarin,jhanyun,wqwang1009,pyaoaa\}@zju.edu.cn}}
    

\maketitle
\begin{abstract}
Program analyzers often issue batches of SMT queries that share a large
symbolic context and differ only in a small predicate. We formalize this
recurring pattern as \emph{Shared-Context Batched Satisfiability}: given a formula $\varphi$ and
predicates $P$, determine whether $\varphi \land p$ is satisfiable for each
$p \in P$.
We study three theory-agnostic strategies for this problem: predicate-by-
predicate checking, disjunctive over-approximation, and Core-Literal Filter
(CLF), a new algorithm that learns literals inconsistent with $\varphi$ and
uses them to reject later predicates.
Our evaluation on symbolic abstraction and active property checking shows that no strategy dominates universally: over-approximation is fastest on solved symbolic-
abstraction queries, while CLF increases the number of solved hard instances and is fastest on active property checking.
We advocate treating shared-context batched satisfiability as a first-class primitive in design program analyzers and exploring the algorithmic design space more systematically.
\end{abstract}


\input{1.intro}
\input{2.background}

\input{3.algo}

\input{4.eval}

\input{5.discuss}

\input{6.related}

\input{7.conclu}


\footnotesize
\bibliographystyle{unsrtnat}
\bibliography{formal,tmp}

\end{document}

%% file: 1.intro.tex
\section{Introduction}
\label{sec:introduction}

Modern \emph{Satisfiability Modulo Theories} (SMT) solvers are the workhorses
behind a wide range of automated reasoning tasks, e.g., deciding the satisfiability for rich theories.
They have become the
de-facto backend for symbolic execution~\cite{sen2005cute,cadar2008exe},
software model checking~\cite{alt2017hifrog,komuravelli2016smt,heizmann2013software}, 
program synthesis~\cite{solar2008program,blazytko2017syntia},
automated repair~\cite{mechtaev2016angelix,nguyen2013semfix}, and refinement
types~\cite{vazou2014refinement,champion2018ice}.  
In the industry, they defend cloud infrastructures and even
validate radiation-therapy machines~\cite{pernsteiner2016investigating,cook2018formal},
with companies like Amazon reportedly issuing \emph{tens of millions} of
SMT queries each day.

When analyzing real-world programs, modern program analyzers issue tens of thousands of satisfiability queries, many of which are highly correlated, such as queries that share a path condition.
In this work, we focus on a recurring computational pattern:
Given a fixed symbolic context~$\varphi$ and a finite set of predicates, determine the satisfiability of $\varphi$ conjoined with each predicate.
This arises, for instance, when verifying multiple properties under a shared path condition, or when checking a family of assertions along a symbolic execution path:
\begin{itemize}
	\item \emph{Active Property Checking}: In symbolic execution, a fixed path condition is conjoined with different property predicates (e.g., division-by-zero, null dereference, buffer overflow) that are checked on the same path. This yields a large number of satisfiability queries sharing a common context.
	\item \emph{Symbolic Abstraction}~\cite{reps2004symbolic,YaoAbst21}: For template linear domains, the best abstract transformers can be computed by encoding program semantics as a fixed formula and solving for parameters that yield over-approximations. OMT solvers explore candidate parameters via iterative refinement, issuing satisfiability queries that vary the predicate while reusing the shared semantics.
\end{itemize}

We isolate and formalize this structure as the problem of \emph{Shared-Context Batched Satisfiability} (SCBS), a descriptive name for the query pattern:
\mybox{
Given a formula $\varphi$ and a set of predicates $P = \{p_1, p_2, \ldots, p_n\}$, determine for each predicate $p_i$ whether the conjunction $\varphi \land p_i$ is satisfiable.
}

There have been related techniques to accelerate similar classes of problems—for example, expression caching in symbolic execution~\cite{visser2012green,jia2015greentrie,aquino2019utopia}, unsat-core reuse in verification~\cite{cimatti2014ic3}, and incremental solving inside SMT solvers.
However, these techniques are designed for broad reuse across heterogeneous queries, such as those arising from different paths, time frames, or analysis stages. 
SCBS workloads are different: \emph{all} queries share the same large context, and the variability is confined to a set of predicates.
While offering general-purpose acceleration, they do not exploit the structural regularities in SCBS.

SCBS workloads possess a unique structural regularity: all queries share an identical, often large, symbolic context, and differ only by a single predicate. This regularity presents two optimization opportunities.
First, when predicates have semantic relationships—implication, mutual exclusion, or logical dominance—classifying one predicate can immediately resolve others without additional solver calls.
Second, a solver invoked on the entire batch can amortize preprocessing, lemma learning, and conflict analysis across all predicates, rather than repeating this work for each query in isolation.

As a first step toward principled support for SCBS, we investigate two existing algorithmic strategies: linear scan and disjunctive over-approximation (\cref{subsec:existing}). We then introduce Core-Literal Filter (CLF), a forbidden-literal–driven algorithm that extracts forbidden literals from unsatisfiable checks to pre-screen subsequent predicates (\cref{subsec:clf}). All methods are theory-agnostic but can incorporate theory-specific enhancements when applicable. We present a detailed analysis and practical optimization heuristics (\cref{subsec:analysis}).
To provide insight into practical behavior, we evaluate the algorithms using queries from active property checking and symbolic abstraction (\cref{sec:eval}). Our results show that each algorithm exhibits distinct performance trade-offs.
We also discuss future directions for enhancing shared-context batched satisfiability (\cref{sec:discuss}).

In summary, we make the following main contributions:
\begin{itemize}
	\item A precise formalization of the shared-context batched satisfiability problem.
	\item A new forbidden-literal–driven algorithm (CLF), which extracts forbidden literals from unsatisfiable checks to pre-screen subsequent predicates.
	\item A systematic analysis and evaluation of three algorithmic approaches with respect to their theoretical properties and practical performance. The implementation is available at \url{https://tinyurl.com/2f99zbun}.
\end{itemize}

%% file: 2.background.tex
\section{Motivating Examples}
\label{sec:overview}
In this section, we illustrate the presence of shared-context batched satisfiability in representative tasks in program analysis.

\subsection{Active Property Checking}
In symbolic bug finding,  active property checking~\cite{godefroid2008active} generalizes runtime verification by checking whether a property holds for all executions that follow a given program path, rather than for a single concrete execution. This is achieved by symbolically executing the path and querying a constraint solver. If the property fails, the solver produces an input that exercises the same path while violating the property.

Concretely, given a symbolic path constraint $pc$, active checkers inject additional constraints encoding potential property violations. For example, a division-by-zero checker introduces the constraint $\phi_{\text{Div}} = (\sigma(d) \neq 0)$, where $\sigma(d)$ denotes the symbolic value of the denominator. The analysis then checks whether $pc \land \neg\phi_{\text{Div}}$ is satisfiable. If so, the resulting model yields a concrete input that triggers a division-by-zero error along the same path.

When multiple properties are checked simultaneously—e.g., buffer overflows, NULL dereferences, integer overflows, and uninitialized variables—each property $P_j$ contributes a constraint $\phi_j$. The analyzer must determine whether $pc \land \neg\phi_j$ is satisfiable for each $j \in {1, \dots, n}$:
$
\text{Check satisfiability of } pc \wedge \neg\phi_j \text{ for each } j \in \{1, \ldots, n\}.$
This is a direct instance of shared-context batched satisfiability, with $\varphi = pc$ and $P = {\neg\phi_1, \dots, \neg\phi_n}$.

\begin{example}\label{ex:active-checking}
\label{example:over}
Consider the following program snippet:
\begin{verbatim}
int process(int x, int *p, int arr[10]) {
    if (x > 5) {
        int result = 100 / (x - 6);     // Divide-by-zero?
        int val = *p + arr[x];          // NULL deref and overflow?
        return result + val;
    }
    return 0;
}
\end{verbatim}

For the path where the condition \texttt{x > 5} holds, the path constraint is $\varphi \equiv (x > 5)$. Three active checkers inject constraints:
\begin{itemize}
    \item Division-by-zero: $\phi_1 = (x - 6 \neq 0)$, i.e., $x \neq 6$.
    \item NULL pointer: $\phi_2 = (p \neq \text{NULL})$.
    \item Array bounds: $\phi_3 = (0 \leq x < 10)$.
\end{itemize}

The analyzer checks three queries, each conjoining $\varphi$ with the negation of one safety condition:
\begin{align*}
    (x > 5) \wedge \neg(x \neq 6) &\equiv (x > 5) \wedge (x = 6) \quad \quad \quad \ \text{\small // SAT, yields } x = 6\\
    (x > 5) \wedge \neg(p \neq \text{NULL}) &\equiv (x > 5) \wedge (p = \text{NULL}) \quad \text{\small //  SAT, yields any } x > 5\\
    (x > 5) \wedge \neg(0 \leq x < 10) &\equiv (x > 5) \wedge (x \geq 10) \quad \quad  \ \ \text{\small // SAT, yields } x \geq 10
\end{align*}
\end{example}

\subsection{Symbolic Abstraction}

Abstract interpretation provides a general framework for sound program analysis by ``executing'' programs over abstract domains that over-approximate concrete program states. The connection between concrete and abstract semantics is formalized through a \emph{Galois connection} $(\alpha, \gamma)$ between a concrete domain $(\mathcal{C}, \leq_\mathcal{C})$ and an abstract domain $(\mathcal{A}, \leq_\mathcal{A})$, where the abstraction function $\alpha: \mathcal{C} \to \mathcal{A}$ maps sets of concrete states to their best abstract representation, and the concretization function $\gamma: \mathcal{A} \to \mathcal{C}$ maps an abstract element back to the set of concrete states it represents. 

A central challenge in abstract interpretation is the design of \emph{abstract transformers}—abstract counterparts of concrete operations that soundly propagate abstract states through program statements. Given a concrete transformer $f: \mathcal{C} \to \mathcal{C}$, any abstract function $f^\sharp: \mathcal{A} \to \mathcal{A}$ satisfying $\alpha \circ f \circ \gamma \leq_\mathcal{A} f^\sharp$ is a sound abstract transformer. Among all sound choices, the \emph{best abstract transformer} $f^\alpha = \alpha \circ f \circ \gamma$ is the most precise. However, this definition is non-constructive—it characterizes the desired result but provides neither a method for computing a representation of $f^\alpha$ nor an algorithm for applying it.

Symbolic abstraction~\cite{YaoAbst21,li2014symbolic,reps2004symbolic,brauer2010automatic} addresses this by encoding program semantics as a formula $\varphi$ and computing the least abstract element $a \in A$ such that $\sem{\varphi} \subseteq \gamma(a)$. In template linear domains, abstract elements are conjunctions of inequalities with fixed linear forms and variable parameters.
The abstraction problem reduces to solving optimization modulo theories (OMT) queries:
$\text{max } {g_1, \ldots, g_n} \text{ s.t. } \varphi$, where each objective $g_i$ is maximized independently.

OMT solvers typically employ iterative refinement. Each iteration evaluates candidate parameter vectors $\vec{c}$ and $\vec{c}'$ by issuing satisfiability queries of the form:

\[
SAT\left(\varphi(\vec{x}, \vec{x}') \land p_{\vec{c}}(\vec{x}) \land \neg p_{\vec{c}'}(\vec{x}')\right) , 
\]
where $\varphi$ encodes the transition relation, and $p_{\vec{c}}$ and $p_{\vec{c}'}$ are predicates instantiated with candidate parameters. Each batch of queries fixes $\varphi$ and varies the predicates, conforming to the SCBS pattern. Multiple objectives may be handled across batches, with some converging early and others resolving later.

\begin{example}
Consider a transition formula $\varphi(x,x') \equiv (x' = x + 1 \land 0 \leq x \leq 10)$ and two template constraints $x' \leq c_1$ and $2x' \leq c_2$. An OMT solver maximizes $c_1$ and $c_2$ by iteratively refining candidate bounds, issuing queries of the form
\[
SAT\!\left(\varphi(x,x') \land (x' \leq c_1) \land (2x' \leq c_2)\right),
\]
repeatedly with the same $\varphi$ but different candidate values of $c_1$ and $c_2$ across iterations. The optimal answers are $c_1 = 11$ and $c_2 = 22$, since $x'$ reaches at most $11$ when $x = 10$. This is precisely the shared-context batched satisfiability pattern: $\varphi$ is fixed across the entire optimization, while the bounding predicates vary per iteration, and the two objectives may converge at different rates. In practice, there may exist a large number of optimization objectives.
\end{example}

%% file: 3.algo.tex
\section{Algorithms for Shared-Context Batched Satisfiability}
\label{sec:existing}

In shared-context batched satisfiability, a fixed formula $\varphi$ is evaluated repeatedly under varying sets of predicates. Moreover, analyzing a single program typically involves solving many such queries. 
This recurring structure lends itself to various optimizations. 
In this section, we first review two existing strategies (\cref{subsec:existing}), then present our Core-Literal Filter (CLF) algorithm (\cref{subsec:clf}), and finally analyze the performance trade-offs of these strategies (\cref{subsec:analysis}).

\subsection{Existing Algorithms}
\label{subsec:existing}

\noindent  \textbf{The Linear Scan Algorithm}.
A straightforward, and the most commonly used, approach is to check whether each predicate $p \in P$ is satisfiable with $\varphi$. Despite its simplicity, the number of solver calls grows proportionally with the number of predicates.
In \cref{subsec:analysis}, we will discuss the optimization of the algorithm, such as solution caching and incremental solving.

\begin{algorithm}[t]
\caption{Predicate-by-predicate check.}
\label{alg:check-one-by-one}
\KwIn{An SMT formula $\varphi$ and a set of predicates $P = \{p_1, \ldots, p_n \}$}
\KwOut{Whether $\varphi \land p_i$ is satisfiable for each $p_i \in P$}
\ForEach{$p \in P$}{
    \If{$\varphi \land p$ is satisfiable}{
        mark $p$ as satisfiable\;
    } \Else{
        mark $p$ as unsatisfiable;
    }
}
\Return marked\_results\;
\end{algorithm}

\smallskip
\noindent \textbf{The Disjunctive Over-Approximation Algorithm}.
To reduce the number of solver queries, Algorithm~\ref{alg:check-overappro}~\cite{DBLP:journals/pacmpl/YaoSHZ21,godefroid2008active} uses a disjunctive over-approximation strategy. The key idea is to aggregate unresolved predicates into a single disjunction and query the solver once for the aggregate, rather than individually for each predicate. This batching amortizes solver cost across multiple candidates while preserving correctness.

\begin{example}
Recall the queries in Example~\ref{example:over}.
A naive approach issues three separate solver calls.  However, a more efficient approach queries the disjunction:  $(x > 5) \wedge (\neg\phi_1 \vee \neg\phi_2 \vee \neg\phi_3)$ once, and extracts from a single model (e.g., $x = 6$, $p = \text{NULL}$) the set of violated properties. This reduces the number of solver invocations and enables early discovery of multiple violations.
\end{example}

\begin{algorithm}[t]
\caption{Compact check via over-approximation.}
\label{alg:check-overappro}
\KwIn{An SMT formula $\varphi$ and a set of predicates $P = \{p_1, \ldots, p_n \}$}
\KwOut{Whether $\varphi \land p_i$ is satisfiable for each $p_i \in P$}


\While{$P \neq \emptyset$} { 
$\Psi \leftarrow \bigvee_{p \in P} p$; \tcp{Create disjunction of remaining predicates}
\label{line:over-appro}
\If{$\varphi \land \Psi \rm \ is \ unsatisfiable$} { \tcp{All remaining predicates are unsatisfiable}
\label{line:over-appro1}
mark every $p \in P$ as unsatisfiable; \\
\Return;
\label{line:over-unsat}
} 
\Else { \tcp{At least one predicate is satisfiable}
$M \leftarrow \text{a model of } \varphi \land \Psi$; \tcp{Get satisfying assignment}
\label{line:over-sat1}
\ForEach{$p_i \in P$}
{
          \If{$M \models p_i$}
			{
			mark $p_i$ as satisfiable\;
			remove $p_i$ from $P$;\tcp{Remove from consideration}
				\label{line:over-sat2}
			}
}
}
}
\Return marked\_results; 
\end{algorithm}

Algorithm~\ref{alg:check-overappro} embodies this strategy and can be viewed as a coarse-to-fine search over the predicate set. Each iteration asks whether \emph{any} predicate in the current pool can be satisfied together with $\varphi$, and uses a model to discharge as many predicates as possible at once. This provides a principled way to reduce solver calls while preserving exactness.
Given a formula $\varphi$ and a set of predicates $P = \{ p_1, \ldots, p_n \}$, the algorithm proceeds iteratively, refining $P$ until all satisfiability results are determined. 
\begin{enumerate}
    \item \emph{Predicate Aggregation (Line~\ref{line:over-appro}):}  
    At each iteration, the algorithm constructs an over-approximated formula $\Psi = \bigvee_{p \in P} p$ by taking the disjunction of all predicates in the current set $P$.

    \item \emph{Satisfiability Check (Lines~\ref{line:over-appro1}--\ref{line:over-sat2}):}  
    The algorithm then checks whether the conjunction $\varphi \land \Psi$ is satisfiable.
    \begin{itemize}
        \item If $\varphi \land \Psi$ is \emph{unsatisfiable}, this implies that no predicate in $P$ can be satisfied alongside $\varphi$. In this case, each conjunction $\varphi \land p_i$ is marked as unsatisfiable, and this round of the algorithm terminates (Line~\ref{line:over-unsat}).
        There is no need for explicit separation of SMT calls for each predicate in $P$.
     \item 
    If $\varphi \land \Psi$ is satisfiable, the SMT solver returns a model $M$. The algorithm iterates through each predicate $p_i \in P$ and evaluates whether $M \models p_i$.
     If a predicate $p_i$ is satisfied by $M$, then $\varphi \land p_i$ is marked as satisfiable and removed from $P$ (Line~\ref{line:over-sat2}).
    \end{itemize}

    \item 
    The process repeats, with $P$ shrinking after each iteration.  The algorithm terminates when all the predicates are marked.
\end{enumerate}

\begin{theorem}
\label{thm:overapprox-bound}
Let $k$ be the number of predicates in $P$ for which $\varphi \land p_i$ is satisfiable. Algorithm~\ref{alg:check-overappro} requires at most $\min(k + 1, n)$ SMT solver calls~\cite{DBLP:journals/pacmpl/YaoSHZ21}.
\end{theorem}

\begin{example}\label{ex:over-approx}
Consider $\varphi \equiv (x \geq 1 \land x \leq 5)$ and predicates $P = \{p_1: x = 2, p_2: x > 3, p_3: x < 0, p_4: x = 4, p_5: x \leq 5\}$. 
\begin{itemize}
    \item First, with $P = \{p_1, p_2, p_3, p_4, p_5\}$, the algorithm constructs disjunction $\Psi = (x = 2) \lor (x > 3) \lor (x < 0) \lor (x = 4) \lor (x \leq 5)$. The solver returns model $M_1 = \{x = 2\}$ for $\varphi \land \Psi$, which satisfies $p_1$ and $p_5$. These predicates are marked as satisfiable. 
    \item Second, with $P = \{p_2, p_3, p_4\}$, the new disjunction yields model $M_2 = \{x = 4\}$, which satisfies $p_2$ and $p_4$. These are marked satisfiable and removed.
    \item Finally, only $p_3: x < 0$ remains. The formula $\varphi \land (x < 0)$ is unsatisfiable since $\varphi$ requires $x \geq 1$. The algorithm marks $p_3$ as unsatisfiable and terminates.
\end{itemize}

Overall, the algorithm requires three SMT calls, versus five for a naive linear scan (a 40\% reduction).
However, the trade-off is that each call solves a larger query---\(\varphi\) conjoined with a disjunction over the remaining predicates---which can increase per-call solving cost and, in the worst case, offset the reduction in the number of solver calls.
\end{example}

\subsection{The Core-Literal Filter Algorithm}
\label{subsec:clf}
Both the linear scan and the over-approximation algorithm treat each unsatisfiable result as an isolated event. We propose a new algorithm, \emph{Core-Literal Filter} (\textbf{CLF}), which extracts reusable information from every unsatisfiable result and propagates it forward to avoid redundant solver calls.

\smallskip
\noindent\textbf{Key insight.}
When $\varphi \land p_i$ is unsatisfiable, some top-level literal $\ell$ inside $p_i$ may itself be inconsistent with $\varphi$, i.e., $\varphi \models \neg\ell$. If so, any future predicate $p_j$ that contains $\ell$ as a top-level conjunct is immediately unsatisfiable without any additional solver call involving $p_j$. The algorithm maintains a \emph{forbidden literal set} $\mathcal{F}$: a set of literals $\ell$ for which $\varphi \models \neg\ell$ has been confirmed. Membership in $\mathcal{F}$ is verified by checking $\varphi \land \ell$ under the precondition alone, which is a once-per-literal cost that can then amortize over arbitrarily many future predicates.

\begin{algorithm}[!h]
\caption{Core-Literal Filter (CLF).}
\label{alg:clf}
\KwIn{An SMT formula $\varphi$ and a set of predicates $P = \{p_1, \ldots, p_n \}$}
\KwOut{Whether $\varphi \land p_i$ is satisfiable for each $p_i \in P$}

$\mathcal{F} \leftarrow \emptyset$; \tcp{forbidden literal set}

Open an incremental solver $\mathcal{S}$ and assert $\varphi$\;

\ForEach{$p_i \in P$}{

    \tcp{Step 1 — free pre-screening}
    \If{\rm $\mathcal{F} \neq \emptyset$ and $p_i$ contains a top-level literal $\ell \in \mathcal{F}$}{
        \label{line:clf-screen}
        mark $p_i$ as unsatisfiable; \tcp{zero solver calls}
        \textbf{continue}\;
    }

    \tcp{Step 2 — incremental satisfiability check}
    $\mathcal{S}.\mathtt{push}()$\;
    assert $p_i$ in $\mathcal{S}$\;
    \label{line:clf-check}

    \eIf{$\mathcal{S}$ is satisfiable}{
        mark $p_i$ as satisfiable\;
        $M \leftarrow$ model from $\mathcal{S}$\; \tcp{Step 3 — model reuse}
        \label{line:clf-reuse-start}
        \ForEach{unresolved $p_j$ with $j > i$}{
            \If{$M \models p_j$}{
                mark $p_j$ as satisfiable\;
            }
        }
        \label{line:clf-reuse-end}
        $\mathcal{S}.\mathtt{pop}()$\;
    }{
        mark $p_i$ as unsatisfiable\; \tcp{Step 4 — forbidden literal extraction}
        $\mathcal{S}.\mathtt{pop}()$\; \tcp{restore to $\varphi$-only scope before verifying literals}
        \ForEach{top-level literal $\ell$ of $p_i$ with $\ell \notin \mathcal{F}$}{
            \label{line:clf-verify-start}
            \If{$|\mathcal{F}| \ge B$}{\textbf{break}\;}
            $\mathcal{S}.\mathtt{push}()$\; assert $\ell$ in $\mathcal{S}$\;
            \If{$\mathcal{S}$ is unsatisfiable under $\ell$}{
                $\mathcal{F} \leftarrow \mathcal{F} \cup \{\ell\}$\;
            }
            $\mathcal{S}.\mathtt{pop}()$\;
            \label{line:clf-verify-end}
        }
    }
}
\Return marked\_results\;
\end{algorithm}

\noindent\textbf{The Algorithm.}
Our algorithm proceeds through the predicates one by one within a single incremental solver session initialized with $\varphi$. For each predicate $p_i$, four steps are executed in order.

\begin{enumerate}
  \item \emph{Free pre-screening (Line~\ref{line:clf-screen}).} Before issuing any solver call, the algorithm checks whether $p_i$ contains a top-level conjunctive literal $\ell$ that belongs to $\mathcal{F}$.  If so, $p_i$ is immediately classified as unsatisfiable at zero additional solver cost, because $\varphi \models \neg\ell$ has already been confirmed.

  \item \emph{Incremental satisfiability check (Line~\ref{line:clf-check}).} If $p_i$ passes pre-screening, the algorithm asserts $p_i$ using \texttt{push}/\texttt{pop} and queries the solver. This is identical to the \emph{Inc} variant of the linear scan.

  \item \emph{Model reuse (SAT case, Lines~\ref{line:clf-reuse-start}--\ref{line:clf-reuse-end}).} When the check is satisfiable, the returned model $M$ is propagated forward: any unresolved $p_j$ ($j > i$) satisfied by $M$ is immediately marked satisfiable, analogous to the \emph{Reuse} heuristic.

  \item \emph{Forbidden literal extraction (UNSAT case, Lines~\ref{line:clf-verify-start}--\ref{line:clf-verify-end}).} When the check is unsatisfiable, the algorithm extracts the top-level conjunctive literals of $p_i$ and tests each candidate $\ell$ against $\varphi$ alone.  A literal $\ell$ for which $\varphi \land \ell$ is itself unsatisfiable is added to $\mathcal{F}$ and can then eliminate future predicates for free. To bound overhead, the size of $\mathcal{F}$ is capped at a constant budget $B$ (a tunable implementation parameter).
\end{enumerate}

\begin{example}\label{ex:clf}
Let $\varphi \equiv (x \ge 0 \land x \le 3)$ and consider the predicate set
\[
P = \{\,p_1{:}\ x < 0,\;
       p_2{:}\ (x < 0)\land(x = 2),\;
       p_3{:}\ x = 1,\;
       p_4{:}\ (x < 0)\land(x = 3)\,\}.
\]
Processing $p_1$ shows $\varphi \land (x < 0)$ is unsatisfiable; moreover, the literal $x<0$ is itself inconsistent with $\varphi$. We therefore add $x<0$ to the forbidden set $\mathcal{F}$ (2 solver calls). Predicates $p_2$ and $p_4$ are then rejected syntactically, since they contain the forbidden literal $x<0 \in \mathcal{F}$ (0 calls each). Finally, $p_3$ is satisfiable under $\varphi$, requiring one solver call.

In total, the procedure uses 3 solver calls for 4 predicates, versus 4 calls for a plain linear scan. Intuitively, the forbidden literal $x<0$ is proved once and subsequently prunes two additional predicates at no solver cost.
\end{example}

\begin{theorem}
\label{thm:clf-bound}
Let $n = |P|$, $k$ be the number of satisfiable predicates, and $B$ the forbidden literal budget. Algorithm~\ref{alg:clf} issues at most $n + (n - k) \cdot B$ SMT solver calls in the worst case, and as few as $1$ call in the best case.  In the common regime where $|\mathcal{F}|$ saturates quickly, each predicate blocked by pre-screening saves one solver call at zero cost.
\end{theorem}

\begin{proof}
Each predicate $p_i$ that is not pre-screened triggers exactly one incremental check in Step 2. For each UNSAT result, at most $B$ additional verification calls are issued in Step 4 to populate $\mathcal{F}$. Since there are at most $n - k$ unsatisfiable predicates, the total verification calls are bounded by $(n - k) \cdot B$. Step 3 (model reuse) and Step 1 (pre-screening) incur no solver calls. Hence, the total is at most $n + (n - k) \cdot B$. If $\bigwedge P$ is satisfiable together with $\varphi$, the very
first incremental check suffices, and all remaining predicates are resolved by model reuse, giving a single solver call.
\end{proof}

\subsection{Algorithm Comparison}
\label{subsec:analysis}

This subsection contrasts the three shared-context batched satisfiability algorithms
discussed so far.
Let $n=|P|$ be the number of predicates and  $k\,(0\le k\le n)$ the number of \emph{satisfiable} conjunctions $\varphi\land p_i$. All complexity figures count \emph{solver invocations}.

\smallskip
\noindent \textbf{Optimization Heuristics}.
The algorithms can benefit from standard SMT optimizations that are orthogonal and can be combined.

\begin{itemize}
\item \emph{Incremental Solving (Inc)}: Use \texttt{push}/\texttt{pop} commands to avoid reasserting $\varphi$ and to trigger the internal incremental solving capability of SMT solvers.  All three algorithms can benefit from Inc; the advantage is most pronounced for Linear Scan and CLF, where each query adds only a single predicate (or literal), making the delta small relative to the shared context $\varphi$.
\item \emph{Model Reuse (Reuse)}: Reuse satisfying models to classify multiple predicates without additional solver calls. Both OA and CLF naturally employ this heuristic: OA discharges multiple predicates per iteration, while CLF propagates each SAT model forward to resolve subsequent predicates.
\end{itemize}

\noindent \textbf{Algorithmic Trade-offs.}
The three algorithms differ in abstraction strategy and performance characteristics. Table~\ref{tab:comparison} summarizes these trade-offs.

\begin{itemize}
    \item \emph{Linear Scan (LS) (Algorithm~\ref{alg:check-one-by-one})}: 
    The algorithm issues \emph{one} query per predicate.  
   Its chief drawback is that the solver may not capitalize on similarities
   between predicates, even with incremental solving. When $n$ is large, the algorithm may become prohibitively slow.

    \item \emph{Over-Approximation (OA) (Algorithm~\ref{alg:check-overappro})}:
    This algorithm 
    excels when only a few predicates are satisfiable, since unsatisfiable disjunctions can eliminate large sets at once. Its performance also depends on the ``quality'' of the models returned: a model that satisfies many predicates allows for rapid pruning, whereas poor models may require more iterations.
    \item \emph{Core-Literal Filter (CLF) (Algorithm~\ref{alg:clf})}:
    The algorithm is effective when predicates share top-level literals that are inconsistent with $\varphi$, allowing forbidden literals to pre-screen many predicates at zero cost. When such sharing is limited, the per-literal verification overhead may reduce the advantage.
\end{itemize}

\begin{table*}[t]
\centering
\caption{Comparison of shared-context batched satisfiability algorithms.}
\label{tab:comparison}
\begin{tabular}{llll}
\hline
\textbf{Aspect} & \textbf{Algorithm~\ref{alg:check-one-by-one}} & \textbf{Algorithm~\ref{alg:check-overappro}} & \textbf{Algorithm~\ref{alg:clf}} \\
\hline
\textbf{Strategy} & Check each $\varphi \land p_i$ & Check $\varphi \land \bigvee P$ & Filter by forbidden literals \\
\textbf{Worst case} & $n$ & $\min(k+1, n)$ & $n + (n-k)\cdot B$ \\
\textbf{Best case} & $n$ & $1$ & $1$ \\
\hline
\end{tabular}

\end{table*}

While theoretical bounds provide useful guidance, empirical performance depends heavily on solver heuristics, predicate structure, and application-specific characteristics.  Given the influence of domain-specific factors on algorithm performance, we present an empirical evaluation in the next section to provide insights into their practical behavior.

%% file: 4.eval.tex
\section{Evaluation}
\label{sec:eval}

We evaluate the algorithms on two real-world clients that generate SCBS queries: \emph{symbolic abstraction} and \emph{active property checking}. Our evaluation addresses three research questions:
\begin{itemize}
    \item \textbf{RQ1}: How do the (relatively) best variants of the three algorithms compare in terms of performance? (\cref{sec:rq1})
    \item \textbf{RQ2}: What are the benefits of incremental solving and model reuse? (\cref{sec:rq2})
    \item \textbf{RQ3}: How do predicate properties and client characteristics influence algorithm selection? (\cref{sec:rq3})
\end{itemize}


\begin{table}[t]
    \centering
    \caption{Summary of evaluated algorithm variants.}
    \label{tab:alg-variants}
    \begin{tabular}{lll}
        \toprule
        \textbf{Variant} & \textbf{Base Algorithm} & \textbf{Optimizations} \\
        \midrule
        LS-Naive     & Algorithm~\ref{alg:check-one-by-one} & None \\
        LS-Inc       & Algorithm~\ref{alg:check-one-by-one} & Incremental solver \\
        LS-Reuse     & Algorithm~\ref{alg:check-one-by-one} & Model reuse \\
        LS-IncReuse  & Algorithm~\ref{alg:check-one-by-one} & Incremental solver + Model reuse \\
        OA           & Algorithm~\ref{alg:check-overappro} & Model reuse \\
        OA-Inc       & Algorithm~\ref{alg:check-overappro} & Incremental solver + Model reuse \\
        CLF           & Algorithm~\ref{alg:clf} & Incremental solver + Model reuse + Pre-screen \\
        \bottomrule
    \end{tabular}
\end{table}

Table~\ref{tab:alg-variants} summarizes the evaluated configurations.
For LS, we explore all combinations of incremental solving and model reuse, resulting in four variants. For OA and CLF, we evaluate only the configurations shown; a full ablation of their internal heuristics is left to future work. For CLF, we set the forbidden-literal budget $B = 32$ for symbolic abstraction and $B = 16$ for active property checking; the impact of this parameter is discussed in \cref{sec:discuss}.

\begin{table*}[t]
    \centering
    \caption{Characteristics of benchmark sets. ``Solved'' counts queries where all solver calls complete within the timeout; ``Total'' includes queries with a timeout.
    }
    \label{tab:benchmark-characteristics}
    \begin{tabular}{lccccc}
        \toprule
        \textbf{Client} & \textbf{\#Solved} & \textbf{\#Total} & \textbf{Avg. SAT Ratio} & \textbf{Predicate Range} \\
        \midrule
        Symbolic Abstraction     & 1,656 & 3,400 & 0.88 & 20--58 \\
        Active Property Checking & 9,153 & 9,200 & 0.89 & 13--27 \\
        \bottomrule
    \end{tabular}
   \vspace{-3mm}
\end{table*}

\begin{itemize}
    \item \textbf{Symbolic Abstraction}. Generates SCBS queries during predicate refinement for verifying real-world programs. We evaluate on ten programs: \texttt{perlbmk}, \texttt{glusterd}, \texttt{gap}, \texttt{eon}, \texttt{wrk}, \texttt{tmux}, \texttt{openssl}, \texttt{vortex}, \texttt{darknet}, and \texttt{transmission}. The high SAT ratio (0.88) reflects the iterative nature of refinement, where a significant fraction of candidate predicates remain feasible.

    \item \textbf{Active Property Checking}. Checks multiple bug detection properties on the same symbolic execution path. Properties include: buffer overflow (CWE-401), double free (CWE-415), use-after-free (CWE-416), null dereference (CWE-476), and additional checks for divide-by-zero, integer overflow/underflow, and uninitialized reads. We evaluate on 27 datasets spanning both CWE-specific benchmarks and real-world programs, including \texttt{darknet}, \texttt{git}, \texttt{htop}, \texttt{lighttpd}, \texttt{mariadb}, \texttt{memcached}, \texttt{tmux}, \texttt{transmission}, \texttt{wget}, and others. The high SAT ratio (0.89) indicates that many property checks are satisfiable, corresponding to potential bugs in the code.
\end{itemize}

All experiments use Z3 version 4.14.1 with a 30-second timeout on a server running Ubuntu 20.04 LTS and equipped with 2 TB of RAM.

\subsection{Comparison of the Best Algorithmic Variants (RQ1)}
\label{sec:rq1}

We compare the strongest representative of each algorithm family: LS-IncReuse for literal search, OA-Inc for over-approximation, and CLF for conflict-literal filtering. These three variants consistently outperform their simpler counterparts; \cref{sec:rq2} analyzes the contribution of the underlying optimizations. We first present per-client results and then summarize the overall tradeoff.

\begin{table}[t]
    \centering
    \caption{Symbolic Abstraction: Average runtime and solver invocations are computed over the 1,656 queries where all algorithms complete within the timeout; timeout rate is computed over all 3,400 queries and reports the fraction in which at least one solver call exceeds the 30-second limit.}
    \label{tab:results-sa}
    \begin{tabular}{lccccc}
        \toprule
        \textbf{Algorithm} & \textbf{Runtime (s)} & \textbf{Time $\downarrow$ (\%)} & \textbf{\#Calls} & \textbf{\#Call $\downarrow$ (\%)} & \textbf{Timeout (\%)} \\
        \midrule
        LS-Naive    & 124.78 &   --    & 44.47 &   --    & 37.6 \\
        LS-Inc      &  22.41 &  82.04  & 44.47 &   0.00  & 26.7 \\
        LS-Reuse    &  24.45 &  80.41  & 12.21 &  72.55  & 32.5 \\
        LS-IncReuse &  11.99 &  90.39  & 12.31 &  72.33  & 25.2 \\
        OA          &  20.16 &  83.84  &  \textbf{7.51} &  \textbf{83.12}  & 31.6 \\
        OA-Inc      &  \textbf{11.57} &  \textbf{90.73}  &  8.56 &  80.75  & 24.8 \\
        CLF         &  13.27 &  89.37  & 27.72 &  37.71  & \textbf{17.5} \\
        \bottomrule
    \end{tabular}
   \vspace{-3mm}
\end{table}

\smallskip 
\noindent \textbf{Symbolic Abstraction}.
Table~\ref{tab:results-sa} shows the results.
 LS-Naive is the clear baseline bottleneck: it requires 124.78~s on average across the 1,656 solved queries and makes 44.47 solver calls per query because it checks predicates sequentially, without either incremental solving or reuse. In contrast, OA-Inc achieves the fastest solved-query runtime at 11.57~s (90.7\% reduction), followed closely by LS-IncReuse at 11.99~s (90.4\%). CLF is slightly slower at 13.27~s (89.4\%). OA achieves the lowest average number of solver calls (7.51, or 83.1\% fewer than LS-Naive), while CLF trades additional calls for its filtering step.

The main distinction of symbolic abstraction is robustness rather than solved-query speed. CLF has the lowest timeout rate at 17.5\%, whereas both LS-IncReuse and OA-Inc time out on about 25\% of the 3,400 total queries. Equivalently, CLF solves 82.5\% of the benchmark, compared with 75.2\% for OA-Inc and 74.8\% for LS-IncReuse. This gap is substantial in absolute terms: CLF solves 272 queries that OA-Inc does not, while OA-Inc solves only 23 that CLF does not. The high timeout rate of this client is driven by instance-level difficulty: 27 of the 60 verification instances time out throughout refinement, whereas 17 are entirely valid.

\begin{table}[t]
    \centering
    \caption{Active Property Checking: Average runtime (ms) per query by algorithm and dataset, computed over solved queries. Best results are in \textbf{bold}.}
    \label{tab:results-apc}
    {%
    \begin{tabular}{lcccc}
        \toprule
        \textbf{Dataset} & \textbf{\#Queries} & \textbf{LS-IncReuse} & \textbf{OA-Inc} & \textbf{CLF} \\
        \midrule
        CWE-416        &  52  & \textbf{6.3} & 6.8  & 6.4 \\
        pbzip2         &  38  & 12.2         & 19.2 & \textbf{9.4} \\
        elfedit        &  30  & 16.4         & 30.8 & \textbf{13.0} \\
        memcached      & 130  & 31.1         & 65.4 & \textbf{26.6} \\
        pesI           & 171  & 78.4         & \textbf{47.5} & 71.9 \\
        lighttpd       & 137  & 25.1         & 36.8 & \textbf{21.8} \\
        htop           & 305  & 22.3         & 33.4 & \textbf{21.4} \\
        tar            & 323  & 28.2         & 36.1 & \textbf{21.7} \\
        wget           & 395  & 43.4         & 45.7 & \textbf{34.7} \\
        darknet        & 1187 & \textbf{28.5} & 36.6 & 30.2 \\
        git            & 1295 & 23.2         & 33.0 & \textbf{20.2} \\
        libicuuc       & 1858 & 23.6         & 29.0 & \textbf{20.9} \\
        \midrule
        Avg.\ (27 datasets) & -- & 28.1 & 42.4 & \textbf{26.5} \\
        \bottomrule
    \end{tabular}
    }
    \vspace{-3mm}
\end{table}

\smallskip 
\noindent \textbf{Active Property Checking}.
Table~\ref{tab:results-apc} presents the active property checking results. Here, the ranking changes: CLF achieves the best overall average runtime at 26.5~ms, outperforming LS-IncReuse (28.1~ms) and OA-Inc (42.4~ms), and it is the fastest method on 19 of the 27 datasets. LS-IncReuse wins on 7 datasets and OA-Inc on only 1. Unlike symbolic abstraction, timeouts are negligible for this client (47 out of 9,200 queries), so average runtime is the dominant metric. Under that metric, CLF is the most effective choice.

\begin{figure}[t]
    \centering
    \includegraphics[width=\linewidth]{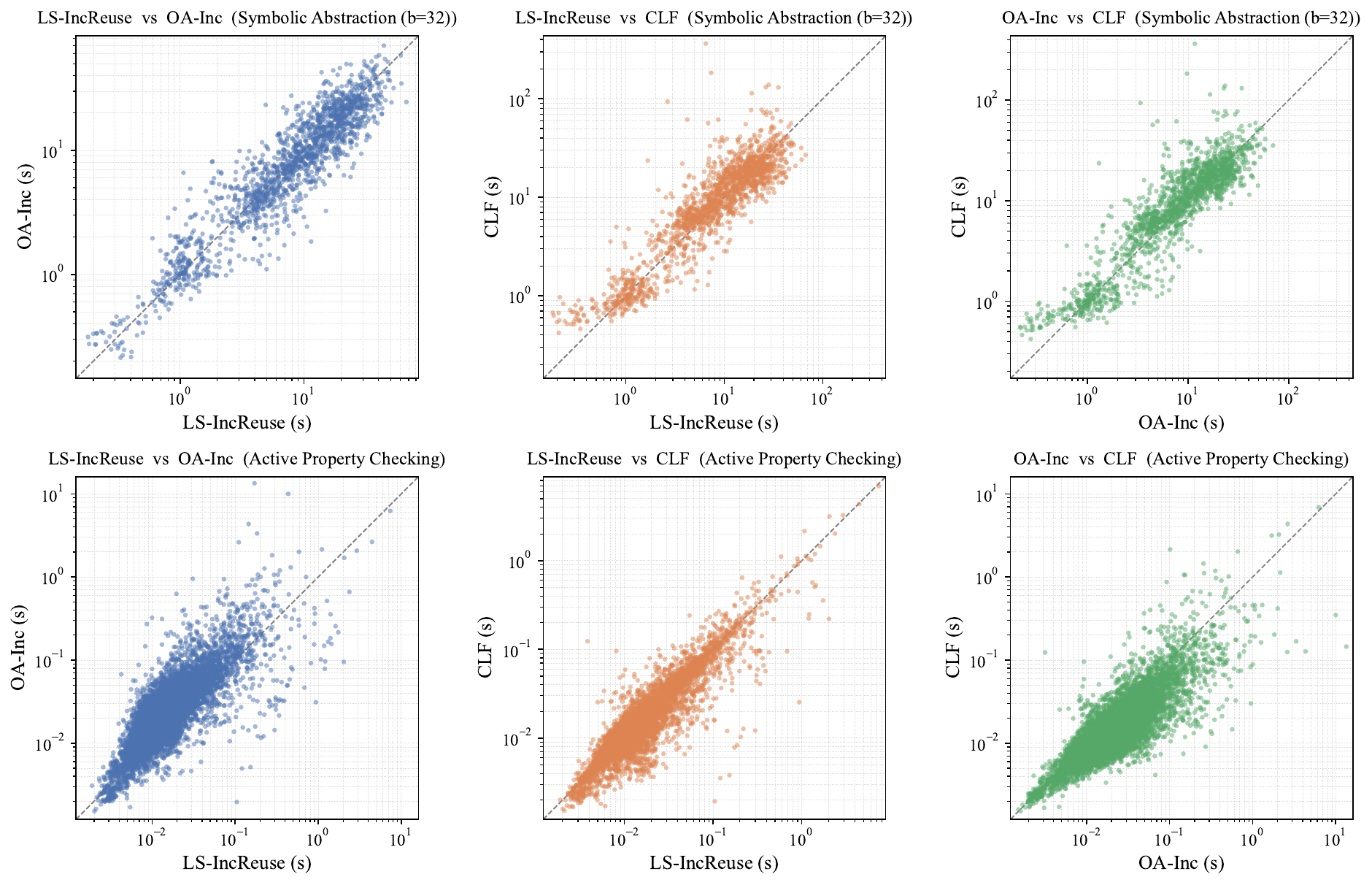}
    \caption{Algorithm comparison scatter plots.}
    \label{fig:rq1-scatter}
     \vspace{-3mm}
\end{figure}

\noindent \textbf{Comparison across Clients}.
Figure~\ref{fig:rq1-scatter} reports the per-query comparison across both clients. The scatter plots confirm that performance is regime-dependent: no single algorithm dominates. On symbolic abstraction, LS-IncReuse and OA-Inc are consistently competitive, particularly on instances solved by all methods. In contrast, CLF is uniformly stronger on active property checking and exhibits greater robustness on the harder symbolic-abstraction instances.

Across both clients, the optimized variants expose a stable tradeoff. OA-Inc achieves the lowest runtime on symbolic-abstraction queries that it solves, whereas CLF solves substantially more difficult queries. For active property checking, where timeouts are uncommon, CLF is also the fastest overall. The primary difference between the two clients is not the SAT ratio (high for both), but the predicate count and the resulting timeout risk: CLF's filtering is efficient enough to win on both runtime and completion when predicate sets are small (APC), while with larger predicate sets (SA), CLF's verification overhead makes it slower on solved queries but its robustness gives it a higher completion rate.

\mybox{\textbf{Answer to RQ1:} No single algorithm dominates across all workloads. On symbolic abstraction, OA-Inc achieves the fastest average runtime on solved queries (90.7\% reduction over LS-Naive), with LS-IncReuse close behind (90.4\%); CLF, however, achieves the highest completion rate (82.5\% vs 75.2\%). For active property checking, CLF achieves the best overall average runtime (26.5 ms), outperforming LS-IncReuse and OA-Inc.}

\subsection{Impact of the Optimizations (RQ2)}
\label{sec:rq2}

We examine how incremental solving and model reuse affect the algorithm's efficiency. Figure~\ref{fig:opt-impact} compares optimized variants against their corresponding baselines. Table~\ref{tab:opt-summary} summarizes the speedup from each optimization.

\begin{table}[t]
    \centering
    \caption{Summary of optimization benefits (runtime speedup). All runtimes are averaged over solved queries.}
    \label{tab:opt-summary}
    \begin{tabular}{lccc}
        \toprule
        \textbf{Optimization} & \textbf{Symbolic Abstraction} & \textbf{Property Checking} \\
        \midrule
        LS-Inc vs LS-Naive & 5.57× & 4.94× \\
        LS-Reuse vs LS-Naive & 5.10× & 3.01× \\
        LS-IncReuse vs LS-Naive & 10.41× & 8.07× \\
        OA-Inc vs OA & 1.74× & 1.80× \\
        \bottomrule
    \end{tabular}
   \vspace{-3mm}
\end{table}

\begin{figure}[t]
    \centering
    \includegraphics[width=\linewidth]{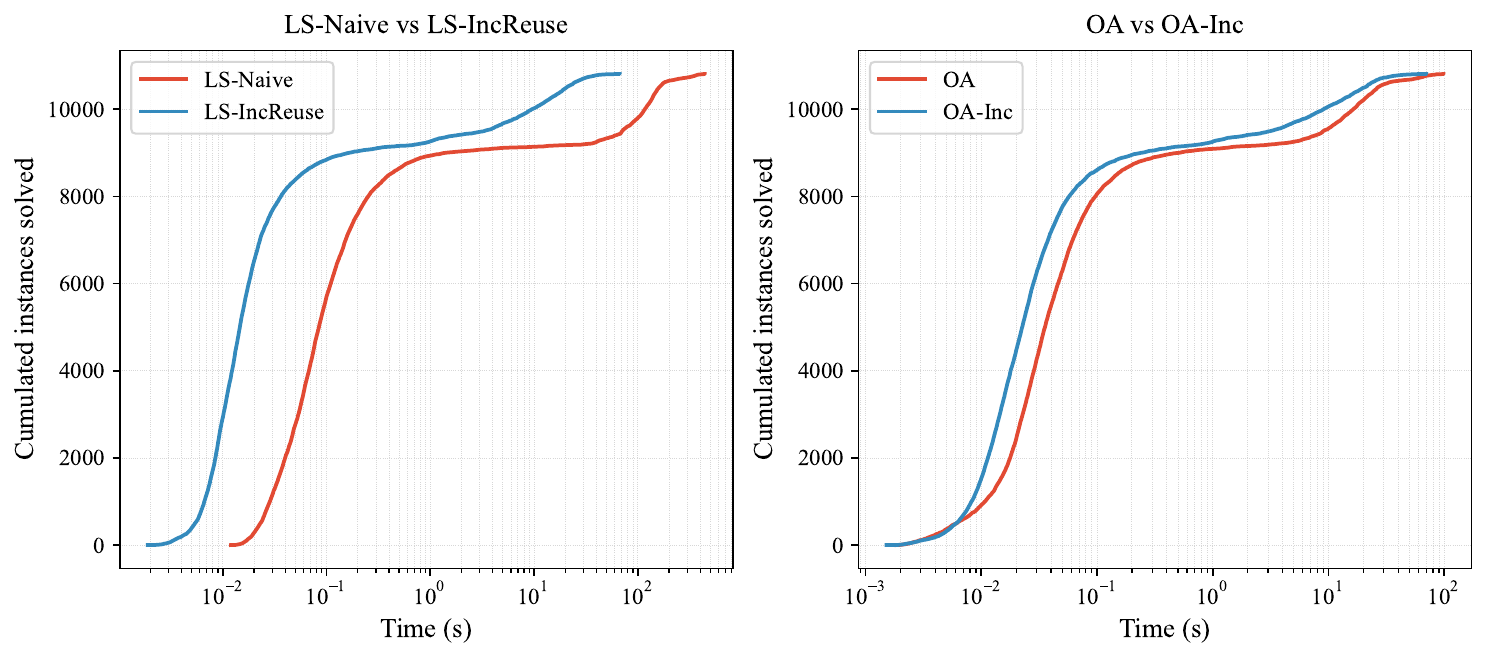}
    \caption{Optimization impact: incremental solving and model reuse.}
    \label{fig:opt-impact}
    \vspace{-3mm}
\end{figure}

\smallskip
\noindent
\textbf{Incremental Solving.}
Solver runtime increases with the satisfiability ratio. On symbolic abstraction, average runtime grows by 1.7--2.2$\times$ from low-SAT (ratio $\le$ 0.85) to near-satisfiable (ratio $\ge$ 0.95) queries, with non-incremental variants (LS-Naive: 2.2$\times$, OA: 2.2$\times$) affected more than incremental ones (LS-Inc: 1.7$\times$, OA-Inc: 1.7$\times$). Notably, CLF is nearly invariant (1.0$\times$), as its pre-screening absorbs the additional cost of high-SAT queries. Algorithm selection also interacts with the SAT ratio. On symbolic abstraction, OA-Inc's win rate drops from 51\% on low-SAT queries (ratio $<$ 0.7) to 30\% on near-satisfiable ones (ratio $\ge$ 0.95), because disjunctive over-approximation provides less pruning when most predicates are satisfiable. During active property checking, CLF achieves the best average runtime on 19 of 27 datasets (most with ratios >0.85), leveraging model reuse to efficiently resolve satisfiable predicates. However, the same advantage does not fully carry over to symbolic abstraction: among the 21 solved SA cases with ratio >0.85, CLF wins only 6---suggesting that factors beyond the SAT ratio, such as predicate structure and count, also play an important role.

\smallskip
\noindent
\textbf{Model Reuse.}
Model reuse significantly decreases solver invocations by caching satisfying assignments that can be projected across multiple predicates. The LS family benefits most: LS-Reuse achieves a 72.5\% reduction in solver calls on symbolic abstraction (63.0\% on property checking). This is particularly effective when SAT ratios are high, as a single satisfying model can resolve many predicates at once. Both OA and CLF incorporate model reuse as a built-in mechanism: OA discharges all predicates satisfied by the current model in each iteration (Algorithm~\ref{alg:check-overappro}, Lines~\ref{line:over-sat1}--\ref{line:over-sat2}), while CLF propagates each SAT model forward to resolve subsequent predicates (Algorithm~\ref{alg:clf}, Lines~\ref{line:clf-reuse-start}--\ref{line:clf-reuse-end}).

\smallskip
\noindent
\textbf{Combined Effect.}
When both optimizations are applied together, the LS family achieves a combined speedup of 10.41$\times$ for symbolic abstraction (8.07$\times$ for property checking), exceeding either optimization alone (5.57$\times$ for incremental solving, 5.10$\times$ for model reuse). This complementary effect arises because incremental solving reduces per-call cost while model reuse reduces the number of calls, targeting different bottlenecks. On symbolic abstraction, LS-IncReuse reduces runtime from 124.78 s to 11.99 s, and solver calls from 44.47 to 12.31.

\smallskip
\noindent
\textbf{Cross-Client Comparison.}
Both optimizations yield larger absolute benefits on symbolic abstraction than on property checking. This difference is driven by predicate count: SA queries involve 20--58 predicates on average, providing more opportunities for model reuse to resolve multiple predicates per solver call. APC queries have fewer predicates (13--27), limiting the absolute number of calls that can be eliminated. Nevertheless, the relative speedup remains substantial across both clients, confirming that these optimizations should be considered mandatory for any practical SCBS implementation.

\mybox{\textbf{Answer to RQ2:} Both incremental solving and model reuse are essential optimizations. For the LS family, where both can be independently evaluated, their combination yields a speedup of 10.41$\times$ on symbolic abstraction and 8.07$\times$ on property checking—exceeding either optimization alone (5.57$\times$ and 5.10$\times$ respectively). Adding incremental solving on top of model reuse provides a further 1.74$\times$ speedup for OA. The two optimizations target complementary bottlenecks and should be considered mandatory for any practical SCBS implementation.}

\subsection{Impact of Problem Characteristics (RQ3)}
\label{sec:rq3}

We evaluate how two key problem characteristics—satisfiability ratio and predicate count—affect solver performance. Figure~\ref{fig:sat-impact} shows runtime as a function of SAT ratio for both clients.
Figure~\ref{fig:count-impact} shows runtime as a function of predicate count. All algorithms exhibit sublinear scaling, though with varying sensitivity.
Table~\ref{tab:client-summary} summarizes the performance characteristics of each client and recommended algorithm selection.

\smallskip
\noindent
\textbf{Satisfiability Ratio}.
As analyzed in \cref{sec:rq2}, the SAT ratio affects both runtime growth and algorithm competitiveness. The implication for algorithm selection is that SAT ratio alone is insufficient to determine the best choice: both clients have similar average ratios (0.88 and 0.89), yet CLF dominates on active property checking (19 of 27 datasets) but wins only 6 of 21 solved SA cases with ratio $>$ 0.85. The decisive factor is whether CLF's pre-screening overhead is amortized by the predicate set—smaller sets (13--27, APC) allow efficient filtering, while larger sets (20--58, SA) incur verification costs that often exceed the savings.

\begin{figure}[t]
    \centering
    \begin{subfigure}[t]{0.49\linewidth}
        \centering
        \includegraphics[width=\linewidth]{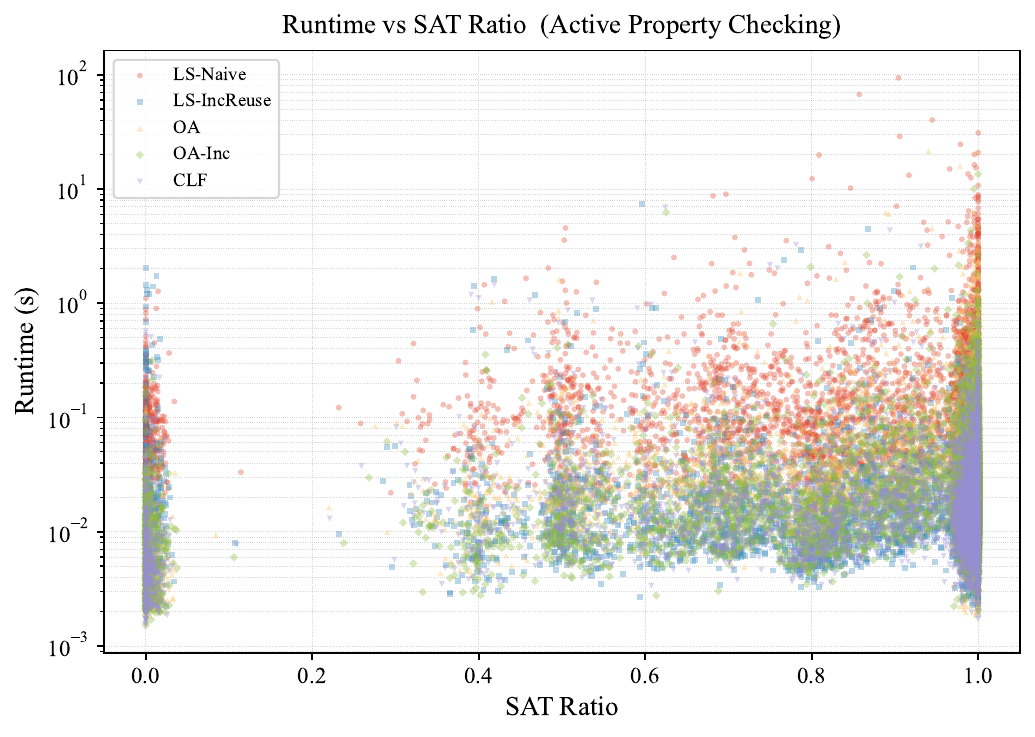}
        \caption{SAT ratio}
        \label{fig:sat-impact}
    \end{subfigure}
    \hfill
    \begin{subfigure}[t]{0.49\linewidth}
        \centering
        \includegraphics[width=\linewidth]{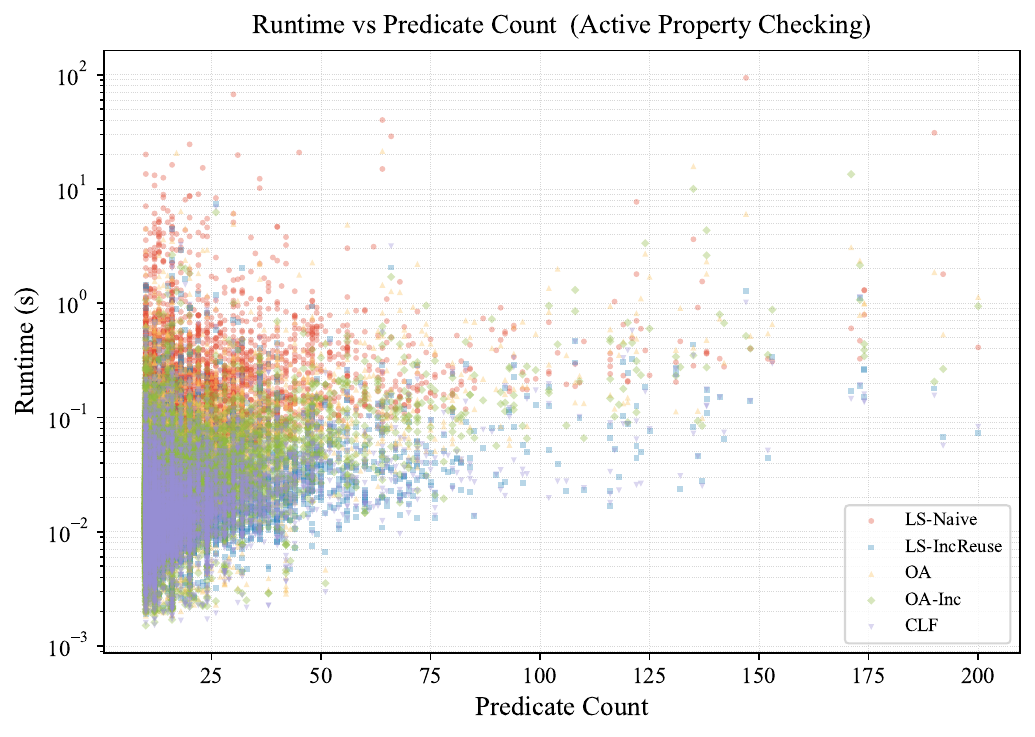}
        \caption{Predicate count}
        \label{fig:count-impact}
    \end{subfigure}
    \caption{Impact of SAT ratio and predicate count on runtime.}
    \label{fig:combined-impact}
    \vspace{-3mm}
\end{figure}

\smallskip
\noindent
\textbf{Predicate Count}.
Runtime increases with predicate count for all algorithms, but the scaling patterns differ. LS-IncReuse and OA-Inc remain closely matched across all predicate counts, with no clear threshold at which one dominates the other. OA-Inc achieves lower solver call counts by discharging several predicates per iteration, while LS-IncReuse compensates with lower per-call overhead from incremental solving. CLF tends to be slower than both on symbolic abstraction instances, as its forbidden-literal verification step adds overhead that grows with predicate count but does not always yield proportionally more filtering.

CLF demonstrates its practical benefit in active property checking, where its model reuse resolves the majority of predicates with few solver calls, and its forbidden-literal pre-screening eliminates the remaining unsatisfiable predicates without solver invocation. On symbolic abstraction, however, CLF's verification overhead often outweighs its benefits, likely because fewer predicates share a conflict structure that the forbidden-literal mechanism can exploit.

\begin{table}[t]
    \centering
    \caption{Client characteristics and algorithm comparison. Completion rate is the fraction of all queries that an algorithm solves within the timeout.}
    \label{tab:client-summary}
    \begin{tabular}{lcc}
        \toprule
        \textbf{Characteristic} & \textbf{Symbolic Abstraction} & \textbf{Active Property Checking} \\
        \midrule
        Avg. SAT Ratio & 0.88 & 0.89 \\
        Predicate Range & 20--58 & 13--27 \\
        Query Timeout Rate & 51.3\% & $<$0.1\% \\
        \midrule
        Fastest Avg.\ Runtime & OA-Inc & CLF \\
        Highest Completion Rate & 82.5\% (CLF) & $\approx$99.9\% (all) \\
        Best Speedup vs LS-Naive & 10.78× (OA-Inc) & 8.6× (CLF) \\
        \bottomrule
    \end{tabular}
  \vspace{-3mm}
\end{table}

\mybox{\textbf{Answer to RQ3:} SAT ratio and predicate count jointly influence algorithm selection. On active property checking (high-SAT, moderate predicate count, negligible timeouts), CLF achieves the lowest average runtime. On symbolic abstraction (high-SAT, large predicate count, frequent timeouts), OA-Inc achieves the fastest average runtime, but CLF completes significantly more queries and should be preferred when robustness matters. 
}

%% file: 5.discuss.tex
\section{Discussions}
\label{sec:discuss}

\noindent\textbf{Applicability of SCBS.} Although our focus is on formulas of the form $\varphi \land p_i$, the shared-context batched satisfiability (SCBS) problem generalizes to a broader class of program analysis tasks that share a common computational structure.
For example,
$k$-induction~\cite{brain2015safety,beyer2015boosting,jovanovic2016property,alhawi2021verification} extends classical mathematical induction to verify temporal properties of transition systems by considering execution traces of bounded length $k$. When verifying multiple safety properties simultaneously, both the base and inductive steps share a common formula representing the system's transition semantics.
Similarly, in invariant inference, modern guess-and-check techniques routinely generate dozens of candidates. SCBS enables efficient evaluation by batching these candidates.

\smallskip
\noindent\textbf{CLF Budget Parameter.}
The forbidden-literal budget $B$ in CLF controls a trade-off between filtering power and verification overhead. A larger $B$ allows more literals to be added to the forbidden set $\mathcal{F}$, enabling more aggressive pre-screening but incurring additional solver calls for verification (Step~4 of Algorithm~\ref{alg:clf}). On the valid SA queries, increasing $B$ improves CLF's head-to-head win rate against OA-Inc from 36.3\% ($B = 16$) to 42.1\% ($B = 64$). However, a larger $B$ also increases the risk of case-level timeouts: with $B = 32$, 320 SA queries return no results for \emph{all} algorithms because CLF's verification calls exhaust the per-case time budget. We use $B = 32$ for symbolic abstraction (which has larger predicate sets, 20--58, that benefit from more aggressive filtering) and $B = 16$ for active property checking (smaller predicate sets, 13--27, where a smaller $B$ suffices). $B = 16$ is a safe default when the predicate set size is unknown; $B = 32$ or larger may be beneficial when predicate sets are large, and the timeout risk is acceptable.


\smallskip
\noindent\textbf{Future Work}. 
The shared-context batched satisfiability problem raises several directions for future work. On the theory side, tighter complexity bounds remain open, potentially obtainable by exploiting structural properties of the target theory or of the predicate families. On the algorithmic side, theory-specific optimizations (e.g., theory-aware lemma caching) may reduce solver interaction and improve throughput. More robust performance may also require adaptive algorithm selection, for instance via portfolio-style combinations that hedge against solver- and instance-specific variability.
Finally, integrating SCBS more deeply with domain-specific contexts (e.g., symbolic execution) may enable further optimization opportunities and scalability gains in practical applications.

%% file: 6.related.tex
\section{Related Work}
\label{sec:related}

\noindent \textbf{Constraint Caching for SMT}.
SMT solvers are integral to modern verification and synthesis tools, enabling reasoning about theories such as bit vectors, arrays, and linear arithmetic. 
They are widely used in various applications, such as symbolic execution~\cite{sen2005cute,cadar2008exe}, software model checking~\cite{heizmann2013software,komuravelli2016smt,jiang2017block,alt2017hifrog}, program synthesis~\cite{solar2008program,blazytko2017syntia}, automated repair~\cite{mechtaev2016angelix,nguyen2013semfix}, and refinement type systems~\cite{vazou2014refinement,champion2018ice}.
To reduce solver overhead, prior work has explored caching mechanisms to avoid redundant queries~\cite{liu2014comparative,jia2015greentrie,aquino2019utopia,visser2012green,jia2015greentrie}. 
KLEE~\cite{cadar2008klee} caches path conditions and counterexamples in symbolic execution. Green~\cite{visser2012green} caches and simplifies queries over linear arithmetic. 
For example, GreenTrie~\cite{jia2015greentrie} extends this by identifying logical implications to increase cache reuse. 
Utopia~\cite{aquino2019utopia} introduces heuristics such as Sat-delta and Unsat-footprint to distinguish satisfiable and unsatisfiable queries for more effective reuse.
These systems are designed for general-purpose reuse across diverse queries that often arise from different paths or time frames. In contrast, our setting is more structured: we evaluate satisfiability over a fixed formula $\varphi$ conjoined with varying predicates $p_i$. 
We hypothesize that combining client-side optimizations can yield additional performance improvements.


\smallskip
\noindent \textbf{Consequence Finding}.
Consequence finding aims to compute the logical entailments of a formula and is widely studied in deduction, such as the computation of prime implicants~\cite{DBLP:conf/fmcad/DeharbeFBM13}.
In the context of circuit verification, equality inference for Boolean functions~\cite{DBLP:conf/iccad/BermanT89} is a well-established technique; identifying equivalent sub-circuits can significantly reduce the complexity of equivalence checking.
In SMT, congruence closure is a standard method for inferring equalities from conjunctions involving uninterpreted functions~\cite{bradley2007calculus}.
In program analysis, consequence finding manifests in various forms, including quantifier elimination~\cite{backeman2018bit,kapur2006automatically,backeman2018bit,john2011quantifier}, interpolation~\cite{sharma2012interpolants,SFS11,mcmillan2006lazy}, 
and implied equalities~\cite{BerdineB14}.
Shared-Context Batched Satisfiability can be applied to identify consequences of a fixed formula $\varphi$ with respect to a set of candidate predicates. 
This restricted form of consequence finding enables localized reasoning within a fixed context.

\smallskip
\noindent \textbf{Predicate Abstraction}.
Introduced by Graf and Saïdi~\cite{Graf97}, predicate abstraction is a foundational technique in program verification~\cite{ball2001automatic,lahiri2006smt,DBLP:conf/vmcai/GulwaniSV09}, which constructs abstractions of infinite-state systems by tracking the truth values of a selected set of predicates.
Although early tools implemented predicate abstraction directly, modern verification frameworks typically employ refined variants, including lazy abstraction with interpolants~\cite{mcmillan2006lazy} and implicit predicate abstraction~\cite{cimatti2014ic3}
A related body of work investigates symbolic abstraction~\cite{reps2004symbolic}, which seeks the best over-approximation of a formula within a given abstract domain, such as finite-height domains~\cite{reps2004symbolic}, template linear domains~\cite{monniaux2009automatic,brauer2010automatic,YaoAbst21}, polyhedral domains~\cite{thakur2012method}, 
In contrast, shared-context batched satisfiability focuses on determining the satisfiability of a predicate in conjunction with a fixed formula, which can serve as a low-level primitive for designing and implementing other algorithms.

%% file: 7.conclu.tex
\section{Conclusion}

Shared-Context Batched Satisfiability is a recurring computational pattern that arises in various applications, yet it has remained hidden in plain sight---buried in implementation details rather than celebrated as the fundamental primitive it truly is.
This paper formalizes the problem, introduces the CLF algorithm, and empirically compares three strategies.
We advocate a more systematic exploration of the algorithmic design space to uncover structure-aware optimizations.